\documentclass[a4paper,12pt]{article}
\usepackage{a4wide, fullpage, sectsty, hyperref}
\usepackage{amsmath, amssymb, amsfonts, amsthm, thmtools}
\usepackage{cleveref}

\numberwithin{equation}{subsection}

\makeatletter
\def\@seccntformat#1{\csname the#1\endcsname. }
\makeatother

\let\savenumberline\numberline
\def\numberline#1{\savenumberline{#1.}}

\theoremstyle{theorem}
\newtheorem{thrm}{Theorem}[section]
\newtheorem{corl}[thrm]{Corollary}

\newtheorem{prop}[thrm]{Proposition}

\declaretheorem[name=Example,style=remark,qed=$\triangle$,sibling=thrm]{exmp}
\declaretheorem[name=Remark,style=remark,qed=$\Diamond$,sibling=thrm]{rema}

\theoremstyle{definition}
\newtheorem{defn}[thrm]{Definition}

\newcommand{\sech}{\operatorname{sech}}
\newcommand\bd{\mathrm{d}}

\begin{document} 

\rightline{QGaSLAB-15-02}
\vspace{1.8truecm}

\vspace{15pt}


{\LARGE{  
\centerline{Probability Density Functions from the Fisher Information Metric} 
}}  

\vskip.5cm 

\thispagestyle{empty} \centerline{
    {\large  T. Clingman$^{a}$\footnote{\tt tslil.clingman@gmail.com}, Jeff Murugan$^{a,b}$\footnote{\tt jeff.murugan@.uct.ac.za}}
   {\large and Jonathan P. Shock$^{a,b}$\footnote{\tt jonathan.shock@uct.ac.za}} }

\vspace{.4cm}
\centerline{{\it $^{a}$The Laboratory for Quantum Gravity \& Strings,}}
\centerline{{\it Astrophysics, Cosmology \& Gravity Center \&}}
\centerline{{\it Department of Mathematics and Applied Mathematics,}}
\centerline{{\it University of Cape Town,}}
\centerline{{\it Private Bag, Rondebosch, 7700, South Africa}}
\vspace{.4cm}
\centerline{{\it $^{b}$National Institute for Theoretical Physics,}}
\centerline{{\it Private Bag X1,}}
\centerline{{\it  Matieland,} }
\centerline{{\it South Africa } }
\vspace{1.4truecm}

\thispagestyle{empty}

\centerline{\bf ABSTRACT}

\vskip.4cm 

We show a general relation between the spatially disjoint product of
probability density functions and the sum of their Fisher information
metric tensors. We then utilise this result to give a method for
constructing the probability density functions for an arbitrary
Riemannian Fisher information metric tensor. We note further that this
construction is extremely unconstrained, depending only on certain
continuity properties of the
probability density functions and a select symmetry of their domains.\\

\setcounter{page}{0}
\setcounter{tocdepth}{2}

\newpage

\tableofcontents

\setcounter{footnote}{0}

\linespread{1.1}
\parskip 4pt

{}~
{}~

\section{Introduction}

Information geometry is the study of the natural differential
structures which arise on the space of families of probability density
functions. The Fisher information metric defines a notion
of the distance between two particular members of a family of
probability density functions and is the natural measure arising out
of the small change expansion of the Kullback-Liebler divergence
\cite{KL1951}. The existence of such a distance measure is of obvious
utility for answering questions related to, for example, the mutual
information of two systems described by different probability density
functions, the likely error made in approximating one distribution by
another, and even a definition of a gradient descent algorithm
consistent with the differential geometric structure of a probability
space \cite{SIA98}.

The study of information geometry was first expounded upon in detail
by Shun'Ichi Amari and the foundations were laid out in
\cite{SIAbook}. A great deal is now known about the geometric
properties of information manifolds. In particular, given a family of
probability density functions, the associated Fisher information
metric may be stated as a concrete integral (or sum in the case of discrete variables). However, comparatively
little is known about the `reverse' operation. That is, given a
Riemannian metric tensor, what can be said about the family of
probability density functions which are naturally endowed with such a
metric tensor? In this short note we show how one can, in theory, perform
this inverse process and observe that it is far from one-to-one.

Our interest in the subject is not from the point of view of machine
learning or information theory as such. In recent years, a new link
has surfaced between information geometry and the study of space-time
as an emergent phenomenon. Within string theory there has been much
work over the last 15 years in the study of how the dynamics of
interacting gauge theories in the limit of a large number of gauge
degrees of freedom can give rise to emergent spacetimes of a variety
of geometries. The most natural such structure arises out of a
scale-free gauge theory providing, holographically, an anti-de Sitter
space \cite{Maldacena} -- the so-called AdS/CFT
correspondence. Coincidentally, the Euclidean version of anti-de
Sitter space (a hyperbolic geometry) is a geometry which emerges
frequently from a large class of different probability density
functions. Indeed in the construction used by Hitchin \cite{Hitchin},
such a space arises naturally out of symmetry arguments when the
Fisher information metric tensor is computed from the instanton moduli
space in such gauge theories. In \cite{Blau} these two ideas were tied
together, showing how Information Geometry seemed to give a natural
means for calculating emergent geometries in an AdS/CFT context.
Interesting relationships between information geometry, quantum
information and string theory/holography have been studied also in
\cite{Tak}, \cite{Matsueda}, \cite{Rey} and \cite{Heck}.

In what follows, we explore in more detail the link between
information and geometry.

\numberwithin{equation}{subsection}
\section{The Fisher information metric}\label{section1}
\subsection{Families of probability density functions and their associated geometries}
For the purposes of this work, we will assume a narrow definition of a
family of probability density functions.  That is, when we write
`family of probability density functions' we will mean a family of
continuous functions $P_{\theta}:X\rightarrow \mathbb{R}$ for some
domain $X\subset\mathbb{R}^{n}$, parameterised over
$\theta\in M\subset \mathbb{R}^{m}$ (ie. an $m$-parameter family of
distributions). Coordinatizing $X$ by $x=(x^{1},\ldots,x^{n})$ and the
parameter space $M$ by $\theta=(\theta^1,\dots,\theta^m)$, we will
also further require that
$\displaystyle \partial_{a}P_{\theta}:=\frac{\partial
  P}{\partial\theta^a}$
is continuous on $X$ for all $\theta\in M$.  Furthermore, we will also
require that every member of the family be normalised, that is,
\[ (\forall\theta\in M)\ \int_{X} P(x;\theta)\,\bd x=1. \]
All of this may be succinctly restated as $\{P_{\theta}\}$ being a
parametrised family of normalised, continuous functions which changes
`smoothly' over parameter space. Finally, we will refer to $X$ as the
\textit{spatial} domain and $M$ as the \textit{parametric} domain, and
conventionally associate the spatial domain $X_{i}$ to probability
density function $P_{i}$.

\newpage

We now define the Fisher Information metric tensor on a finite
dimensional statistical manifold. Given such a manifold, ${\cal M}$,
whose points form a family of probability density functions with the
properties listed above, there exists a Riemannian metric tensor on
${\cal M}$, viz.,
\begin{equation}
  g_{ab}(\theta) = \int_X\,P(x;\theta)\,\partial_a\ln P(x;\theta)\,\partial_b\ln
  P(x;\theta)\, \bd x.
\label{eq:Fisher}
\end{equation}
The central question addressed in this paper may thus be stated as:
given a Riemannian metric tensor $g$, under what circumstances can a
family of probability density functions $P$ be found such that the Fisher
information metric tensor of $P$ is $g$.

\subsection{Some examples}
In order to build some intuition for the relationship between a family
of probability density functions and their associated metrics, we give
here two examples of the computation of the Fisher metric.

\subsubsection{Univariate Normal Distribution}
Here the family of probability density functions is given by
\begin{equation*}
  P(x;\theta)=\frac 1 {\sigma\sqrt{2\pi}} e^{-\frac 1 2 \left(\frac{x-\mu}\sigma\right)^2}.
\end{equation*}
The distribution is parameterised by $\mu$ and $\sigma$, which we will
collectively denote $\theta$. Put another way, the manifold
coordinates are given by $\theta=(\mu,\sigma)$, and the random
variable is $x\in \mathbb{R}$. Note that the parametric domain is
$\mathbb{R}\times\mathbb{R}^{>0}$. In order to compute $g_{ab}$ we
must compute $\partial_{a}\ln P$
\begin{align*}
\ln P&= -\left[\frac 1
  2\left(\frac{x-\mu}\sigma\right)^2+\ln\sigma+\ln\sqrt{2\pi}\right], \\
\frac{\partial}{\partial\mu}\ln P &= \frac 1 \sigma \left(\frac{x-\mu}\sigma\right),\quad
\frac{\partial}{\partial\sigma}\ln P = \frac 1 \sigma \left[\left(\frac{x-\mu}\sigma\right)^2-1\right].
\end{align*}
Then, using \autoref{eq:Fisher}, the Fisher metric for the univariate normal distribution has
\[ [g]=
\begin{bmatrix}
  \frac 1 {\sigma^2} & 0 \\
  0 & \frac 2 {\sigma^2}
\end{bmatrix} \implies \bd s^2=\frac{\bd \mu^2+2\bd\sigma^2}{\sigma^2}.
\]
Thus we see that the Fisher metric, in this case, describes the metric
tensor of a two-dimensional hyperbolic geometry. The structure on this
geometry can be intuitively understood by the properties of normal
distributions. In particular, for distributions with $\sigma\gg 1$,
the associated `difference' between two distributions with means
$\mu_1$ and $\mu_2$ is less pronounced -- they are harder to
distinguish. For two sharply peaked distributions ($\sigma\ll 1$) with
even similar $\mu$, the difference will be very pronounced and so they
are easy to distinguish. Hence the hyperbolic nature of the space.

\subsubsection{Cauchy Distribution}
The family of probability density functions for this distribution is
given by
\[ P(x; x_0,\gamma)=\frac 1 \pi\left[\frac \gamma
  {\gamma^2+(x-x_0)^2}\right]. 
\]
Thus, the parameter space for this family is spanned by the parameters
$\theta=(x_0,\gamma)\in \mathbb{R}\times \mathbb{R}^{>0}$ and the
calculation of the logarithmic derivatives gives
\begin{align*}
  \ln P&=\ln\gamma-\ln\left[\gamma^2+(x-x_0)^2\right]-\ln\pi,\\
  \frac{\partial}{\partial x_{0}}\ln P&=\frac {2(x-x_0)}{\gamma^2+(x-x_0)^2},\quad
  \frac{\partial}{\partial\gamma}\ln P=\frac 1 \gamma-\frac{2\gamma}{\gamma^2+(x-x_0)^2}.
\end{align*}
As such, it is a simple matter to verify that the Fisher metric for
the Cauchy distribution is given by
\[ g_{ab}=\frac{\delta_{ab}}{2\gamma^2}\implies \bd s^2=\frac 1
2\left(\frac{\bd x_0^2+\bd\gamma^2}{\gamma^2}\right).
\]
The reader may wish to note that while we started with a very
different distribution, the geometric structure described by its
Fisher metric is very close to that of the normal
distribution. In this sense, hyperbolic spaces (or Euclidean anti de-Sitter spaces) appear
ubiquitous in an information geometric context.

\numberwithin{equation}{subsection}
\section{Reversing the Fisher information metric}

It is not clear at first glance that it is at all possible to reverse
the process of computing the Fisher metric in any meaningful way, as
the exercise involves a definite integral of multiple powers of the
underlying family of probability density functions. We present below a
motivating example to suggest that under certain, constrained
situations such a process is indeed possible. As a prototype for a more general construction, 
we demonstrate how to encode the metric tensor of $\mathbb{S}^{n}$, for any
$n\in \mathbb{N}$, in a family of one dimensional probability density
functions.

\subsection{The $n-$dimensional sphere, $\mathbb{S}^{n}$}

We begin our exploration of reversing the Fisher information
computation with a one-dimensional family of
probability density functions. In particular, we leverage the
properties of orthonormal functions to produce a family of probability
density functions which, with an appropriate set of functions $h^i$, 
give rise to the metric tensor of $S^n$.

Note that, for our purposes, a family of univariate, real-valued
functions $\{f_i(x)\}_{i\in I}$ is said to be orthonormal with weight
$w(x)$ over a domain $X$ if
$\int_X f_i(x)f_j(x)w(x)\bd x=\delta_{ij}$.

\begin{prop}
  \label{prop:1dd}
  Let $M\subset \mathbb{R}^{n}$ and $h^i\in C^{1}(M)$
  such that\footnote{Here we use Einstein
    summation and the lowered and raised indices have no differential
    geometric interpretation other than to aid in the appropriate
    summations} $(\forall\theta\in M)\, h^{i}h^{j}\delta_{ij}=4$
  and $\{f_i(x)\}_1^n$ be a set of orthonormal, real-valued functions
  with positive semidefinite weight $w(x)$ over
  $X\subset\mathbb{R}$. Then the family of probability density
  functions
  \begin{equation}
    \label{eq:1dd}
    P(x;\theta)=\frac14\left(\sum_{i=1}^nh^i(\theta)f_i(x)\right)^2w(x),
  \end{equation}
  gives the Fisher information metric tensor
  $ g_{ab}=(\partial_a h^i)(\partial_bh^j)\delta_{ij}$.
\end{prop}

\begin{proof}
  That $P$ is normalised follows trivially from the orthonormality of
  $f_i$.
  \begin{align*}
    &\frac14\int_X\left(\sum_{i=1}^nh^i(\theta)f_i(x)\right)^2w(x)\,\bd x
    =\frac14\int_X\left(\sum_{i=1}^n\sum_{j=1}^nh^ih^jf_if_j\right)w\,\bd x\\
 &=\frac14\sum_{i=1}^n\sum_{j=1}^n\int_Xh^ih^jf_if_jw\,\bd  x
 =\frac14h^ih^j\delta_{ij}=1.
  \end{align*}
  A straightforward computation gives the desired result.
  \begin{align*}
    g_{ab}&=\int_X P(\partial_a \ln P)(\partial_b\ln P)\bd  x\\
    &=\int_X w\left(\sum_{i=1}^nh^if_i\right)^2
    \left(\frac{\sum\limits_{i=1}^n(\partial_a
        h^i)f_i}{\sum\limits_{i=1}^nh^if_i}\right)
    \left(\frac{\sum\limits_{i=1}^n(\partial_b
        h^i)f_i}{\sum\limits_{i=1}^nh^if_i}\right)\,\bd x\\
    &=\sum_{i=1}^n\sum_{j=1}^n\int_X(\partial_ah^i)(\partial_bh^j)f_if_jw\,\bd 
    x =(\partial_ah^i)(\partial_bh^i)\delta_{ij}.
  \end{align*}
\end{proof}

Now we pause to note that we may view the above statement,
$g_{ab}=(\partial_ah^i)(\partial_bh^i)\delta_{ij}$, as the 
result of applying the transition functions $h$
to the flat Euclidean metric $\delta$. As such, and noting that we required
$h^{i}h^{j}\delta_{ij}=4$, we immediately infer that

\begin{corl}
  The metric tensor of $\mathbb{S}^n$ can be reached as the Fisher
  Information metric of the distribution \autoref{eq:1dd} where $h$ is
  the transition function from $\mathbb{E}^n$ to $4\mathbb{S}^n$, the
  $n$-dimensional sphere of radius four.
\end{corl}

In the above we have shown a general way to find a given metric tensor
in terms of the transition functions from flat Euclidean space to a
desired geometry. However, there is a specific condition on the $h^i$
given by $h^ih_i=4$ which constrains these strongly. In what follows,
we will generalise this result in a way which will remove this
constraint.

\subsection{The Gaussian construction}\label{ssec:gauss}
Now that we have reason to believe that it is possible, at least in
special cases, to pick a metric tensor and construct a family of
probability density functions whose Fisher information metric is the
selected metric, we attempt to extend our results to arbitrary
Riemannian metrics.

Consider a family of probability density functions given by a product
of $n$, uncorrelated, disjoint, one-dimensional Gaussian probability
density functions with unit variance. Explicitly,
\begin{equation}
\label{eq:gaussian1}
 P(x;\theta) = \frac 1 {\sqrt {(2\pi)^n}}
\exp\left(-\frac 1 2 \sum\limits_{i=1}^n\left(x^i- h^i(
      \theta)\right)^2\right),  
\end{equation}
where $M$, the parametric domain, is not yet fixed,
$X=\mathbb{R}^{n}$, and $h^i\in C^{1}(M)$. From this, we may
compute the Fisher information metric as follows
\begin{align*}
  g_{ab}&=\frac 1 {\sqrt{(2\pi)^n}}\int_{X}\bd  x\,
  e^{-\frac 1
    2\sum\limits_{i=1}^n\left(x^i-{h^i}\right)^2}\left(\sum_{j=1}^n(\partial_ah^j)\left(x^j-{h^j}\right)\right)
  \left(\sum_{k=1}^n(\partial_bh^k)\left(x^k- {h^k}\right)\right)\\
  &=\frac 1 {\sqrt{(2\pi)^n}}\int_{X}\bd  x\, e^{-\frac 1
    2\sum\limits_{i=1}^n\left(x^i- {h^i}
      \right)^2}\left(\sum_{j=1}^n(\partial_ah^j)(\partial_bh^j)\left(x^j-
    {h^j} \right)^2+\parbox{5em}{{\centering vanishing\\cross-terms}}\right)\\
  &=\sum\limits_i\left\{(\partial_a h^i)(\partial_b
    h^i)\prod\limits_k\left(\frac 1 {\sqrt {2\pi}}\int\limits_{-\infty}^\infty\bd  x^k\,e^{-\frac
        1 2\left(x^i-{h^i} \right)^2}\left(x^k- {h^k}
        \right)^2\right)\right\}.
\end{align*}
It is a simple matter to complete the computation to obtain
\begin{equation}
  \label{eq:cm}
  g_{ab}=(\partial_ah^j)(\partial_bh^k)\delta_{jk}.
\end{equation}

This result allows us enough flexibility to be able to always give an
$h$ and $M$ such that $g_{ab}$ may be constructed as desired. In particular,
we may begin at \autoref{eq:cm} and read backwards to find
\autoref{eq:gaussian1}. In doing so, we fix a desired $g_{ab}$ and
accompanying manifold $\mathcal{M}$, and attempt to realise an $h$ and
$M$ for which \autoref{eq:cm} would hold. Unlike the case of
\cref{prop:1dd}, which came with the constraint $h^{i}h_{i}=4$, this
process is here always possible.

The Nash Embedding Theorem \cite{Nash} tells us that there is an $n\in \mathbb{N}$
such that $(\mathcal{M},g)$ may be $C^{1}$ isometrically embedded in
$(\mathbb{E}^{n},\delta)$. Specifically then, it tells us that there
exists an $h$ such that $g=h^{*}\delta$. As such, interpreting
\autoref{eq:cm} as the statement that $g$ is the pullback of $\delta$
via $h$ we see that we need only select an $n$ large enough to
accommodate the Nash embedding of the desired manifold $\mathcal{M}$ in
$\mathbb{E}^{n}$ (which is always possible) and we have $h$ and $M$ to
satisfy the arrangement. Consequently, we have a family of probability
density functions, given by \autoref{eq:gaussian1} whose Fisher
information metric is the desired, arbitrary Riemannian metric.

Said another way, \autoref{eq:cm} states simply that $g_{ab}$ is the
pullback from a higher dimensional flat space to a manifold embedded
in that space, via $h$. In the case of coincidence of dimensions
between $g$ and $h$, the result bears the simple interpretation of $h$
acting as a set of transition functions from $\delta$ to $g$.

\subsubsection{The metric of $\mathbb{S}^{2}$}
To cement the understanding of the importance and generality of
\autoref{eq:cm} we construct the metric tensor of
$\mathbb{S}^{2}$. Suppose we desire a family of probability density
functions whose Fisher information metric is the metric tensor of
$\mathbb{S}^{2}$. Specifically, if the unit sphere has line element
\begin{equation*}
ds^2=\bd \theta^2+\sin^2\theta\bd \phi^2,
\end{equation*}
then we can proceed as outlined above, and write down a set of
transition functions
\begin{equation*}
h=(\cos\theta\sin\phi,\sin\theta\sin\phi,\cos\phi),
\end{equation*} 
from $\mathbb{E}^{3}$ to the embedded $\mathbb{S}^2$. Applying the construction
of \autoref{eq:gaussian1} we find

\begin{equation*}
  P(x,y,z;\theta,\phi)={(2\pi)^{-\frac 3 2}}e^{-\frac 1
  2\left(\left(x-{\cos\theta\sin\phi}\right)^2+
    \left(y-{\sin\theta\sin\phi}\right)^2+\left(z-{\cos\phi}\right)^2\right)}.
\end{equation*}

This is easily recognisable as a product of three Gaussian probability
density functions, each with a mean which is periodic in the
parameters. This means that we have the geometry and topology of a
sphere, where each point on the sphere corresponds to a three
dimensional Gaussian distribution with unit variance and mean denoted
by the point on the sphere. This exercise can be performed for any
$\mathbb{S}^n$ by simply forming the appropriate $h$.

The ease with which we are able to perform this construction is
indicative of the power underlying \autoref{eq:cm} and the
accompanying statement that {\it any Riemannian metric tensor may be
reached via this construction}.

\numberwithin{equation}{subsection}
\subsection{The hyperbolic secant construction}\label{ssec:sec}
In the previous subsection we gave a construction based upon a product
of Gaussian probability density functions and demonstrated its
flexibility. Now we demonstrate that the above-mentioned
results are just as achievable with an entirely different family of
probability density functions. Consider the family

\[ P=\frac 1 {\pi^n}\prod\limits_{i=1}^n\sech\left(x^i-h^i\sqrt
  2\right).\]
\newpage
Other than the functional dependence on $\sim x^i-h^i$, this is
entirely different from the Gaussians discussed earlier. However,
computing the Fisher information metric we find the
result to be of that most general form
\[g_{ab}=(\partial_a h^i)(\partial_b h^j)\delta_{ij}.\]
Naturally, this bears the same interpretation as the previous result
and serves to suggest that relatively little of the information about
the original family of probability density functions is carried
through to the metric tensor itself.

The careful reader will note that we now have \textit{two} means to
the same end, and may wonder just how many more ways we may achieve
the above result. Indeed the following section serves to introduce a
general framework which will show that the answer is that there is an
infinite-fold degeneracy in the construction, and thus there is always
an infinite to one mapping between families of PDFs and Riemannian
metrics via the Fisher information metric.

\numberwithin{equation}{section}
\section{General results}\label{section2}
In this section we will elaborate on a more general set of statements
which allow for definitions independent of dimensionality and
functional dependence of the parameters of the PDF in question.
We begin by showing how to construct a family of
spatially disjoint probability density functions out of individual
families of probability density functions.

\begin{defn}
  The spatially disjoint product of two families of probability
  density functions on the same parametric domain,
  $P_1=P_1(x^1,\ldots,x^k;\theta):X_{1}\times M\rightarrow\mathbb{R}$ and\newline
  $P_2=P_2(x^1,\ldots,x^n;\theta):X_{2}\times M\rightarrow\mathbb{R}$, is
  defined as
  \[
  (P_1\odot P_2)(x^1\ldots,x^{n+k};\theta)=
  P_1(x^1,\ldots,x^n;\theta)\cdot
  P_2(x^{n+1},\ldots,x^{n+k};\theta).
  \]
  Note that $P_{1}\odot P_{2}:(X_{1}\times X_{2})\times
  M\rightarrow\mathbb{R}$ and we write $P^{\odot n}$ where we mean
  $\bigodot_{i=1}^{n}P$.
\end{defn}

Given this, we will here show how a special property of spatially
disjoint products underpins all the general results achieved in this
work. That is, the Fisher information metric transforms the spatially
disjoint product of probability density functions into a sum of their
corresponding, individually considered metric tensors.

\begin{thrm}
  \label{thrm:sdp}
  If $P=P({ x};\theta)$ is a probability density function with a
  decomposition\newline $P=\bigodot P_i^{\odot e_i}$ for some $P_i$ and
  $e_i\in\mathbb{N}^{+}$ then
  $g_{ab}\left(\bigodot P_i^{\odot e_i}\right)=\sum e_i g_{ab}(P_i)$.
\end{thrm}

\begin{proof}
  Let us rewrite $P=\bigodot \hat{P}_{i}^{\odot e_{i}}=\bigodot P_j$
  where each $\hat{P}_i$ has been accumulated into the spatially
  disjoint product $e_i$ times, that is, $P_{j}=\hat{P}_{i}$ for
  $e_{i}$ many $j$. Then, in order to compute $g(P)$ we expand
  logarithmic derivatives to arrive at
  \[g_{ab}(P)=\sum_i\sum_j\int_{X}\bd { x}\,\frac{P}{P_i
    P_j}(\partial_a P_i)(\partial_b P_j).\]
  To proceed we must evaluate the double sum, and to do so we examine
  the cases $j=i$ and $j\ne i$ separately. In the event of the latter,
   $j\ne i$, we have
  \[ \int_{X}\bd { x}\,\frac{P}{P_i P_j}(\partial_a P_i)(\partial_b
  P_j)=\left(\int_{X_{i}}\bd\,  x^a\cdots\bd  x^k\partial_a
    P_i\right)\left(\int_{X_{j}}\bd\,  x^m\cdots\bd  x^r\partial_b
    P_j\right), \]
  where we have expanded the integral as a product over its disjoint
  spatial domains and have suppressed all other terms as they were of
  the form $\int_{X_{i}}\bd  x^a\cdots\bd  x^kP_i=1$.  Moreover, we note
  that $P_{i}$ satisfies the conditions (by the definition of the
  probability density function) for the exchange of integral and
  derivative and so
\[ \int_{X_{i}}\bd  x^a\cdots\bd  x^k\partial_{a}P_i
= \partial_{a}\int_{X_{i}}\bd  x^a\cdots\bd  x^k P_i=\partial_{a}(1)=0.
\]
Thus contributions from terms where $j\ne i$ is zero. On the
other hand, the cases for which $i=j$ admit simple resolution as
\[ \int_{X}\bd  {x}\frac{P}{P_iP_i}(\partial_a
P_i)(\partial_b P_i)=\int_{X_{i}}\bd  x^a\cdots\bd  x^k(\partial_a
P_i)(\partial_b P_i) \frac 1 {P_i}=g_{ab}(P_i),
\]
where again we have expanded the integral as a product and suppressed
all terms whose integral was one. Finally, we recall that we had
exactly $e_{i}$ many $P_{j}$ such that $P_{j}=\hat{P}_{i}$ and so we
collect $e_i$ many such contributions of $g_{ab}(P_{j})$.
\end{proof}

\begin{rema}
  That we essentially require $M_{1}=M_{2}=M$ in the definition of the
  spatially disjoint product is a matter of some subtlety. Consider
  that if $M_{1}\ne M_{2}$ we would be within reason to set
  $M=M_{1}\times M_{2}$ and reinterpret the definition as
  \[(P_1\odot P_2)(x^1\ldots,x^{n+k};\theta,\phi)=
  P_1(x^1,\ldots,x^n;\theta)\cdot P_2(x^{n+1},\ldots,x^{n+k};\phi).
  \]
  In this case, however, $g(P)$ is not strictly the sum of $g(P_{i})$
  as the latter may all be of different dimension. Simply
  re-interpreting $P_{i}$ to have enlarged parametric domain $M$ will
  not solve this problem as then it may happen that $g(P_{i})$ will no
  longer be non-degenerate and so not a metric tensor. Thus, the
  direct ability of the above result to ``glue'' together disjoint
  metric tensors is apparent, but nuanced and not an immediate
  consequence of the exposition given.

  In effect then, care should be taken when examining the statement
  $g(\bigodot P_{i})=\sum g(P_{i})$ so as to ensure that it is done
  with the understanding that $g(P_{i})$ is to have zero entries where
  appropriate for the purpose of the sum, but not when considered as
  its own metric tensor. More formally, we could write
  $g(\bigodot P_{i})=\sum \tilde g(P_{i})$ where $\tilde g$ is
  expressed precisely as $g$, but is extended to all of $M$ as
  suggested above, and is free from interpretation as a metric
  tensor. Hereafter, it is taken for granted that such nuances are
  appreciated by the reader.
\end{rema}

The importance of \cref{thrm:sdp} cannot be overstated. From here on,
it is simply a matter of finding convenient forms of $g_{ab}(P_{i})$
for some parameterisation of $P_{i}$ so that we may take
$\bigodot P_{i}$ and arrive at a desired metric tensor. That is, if we
can find a $P_{i}$ such that
$g_{ab}(P_{i})\propto (\partial_{a}h^{i})(\partial_{b}h^{i})$ then we
can take $P=\bigodot P_{i}$ to find
$g_{ab}\propto (\partial_{a}h^{i})(\partial_{b}h^{j})\delta_{ij}$ by
the above. Here, the whole is more than the sum of its parts -- given
$g_{ab}\propto (\partial_{a}h^{i})(\partial_{b}h^{j})\delta_{ij}$ we
are able to find an $h$ for our desired manifold and then create a
desired $P$ out of constituent $P_{i}$, each containing some part of
$\{h^{i}\}$. Beginning with disjoint $P_{i}$, however, the qualities
which the individual distributions should exhibit, to attain a given
$g$, are not clear. Furthermore, we note here that while
$\bigodot P_{i}$ will yield the desired result, if we find multiple
families of probability density functions, we may equally well combine
them to achieve the same result.

Thus, what we really seek are simple forms of functional dependence of
families of probability density functions upon our set of
differentiable functions $h$ so that explicit computations may be
made. Recall that we saw, in the calculations in subsections
\ref{ssec:gauss} and \ref{ssec:sec}, that we may leverage
reparameterisation invariance of spatial domains to our
advantage. Such symmetries of the spatial domain allow us to
essentially eliminate any functional dependence of the integrals upon
the $h^{i}$ and produce multiplicative factors of $\partial_{a}h$ in
the process. To that end, we explore a generalisation of the symmetry
used in the above-mentioned subsections.

\begin{prop}
  \label{prop:symm}
  Fix a one-dimensional probability density function $\hat{P}(x)$ on
  $X$ for which $X$ remains invariant under the change of variables
  $y=f(x;\theta)$, for some differentiable family of diffeomorphisms
  $f:X\times M\rightarrow X$ (the parameter space is $M$) and let
  $P(x;\theta)=f_{x}(x;\theta)\hat{P}(f(x;\theta))$ such that
  $\partial_{a}P\not\equiv 0$ where we write $f_{x}$ for
  $\frac{\partial f}{\partial x}$ and $f_{a}$ for
  $\partial_{a}f$. Then
  \begin{equation}
   \label{eq:symm}\tag{4.4}
 g_{ab}(P)=\int_{X}\frac{f_{ax}f_{bx}}{(f_{x})^{2}}\hat{P}(y)+\left(\frac{\partial(f_{a}f_{b})}{\partial
   y}+f_{a}f_{b}\frac{\bd \ln \hat{P}(y)}{\bd 
   y}\right)\frac{\bd  \hat{P}(y)}{\bd  y}\bd y,    
  \end{equation}
  where we assume that we have written all functions in terms of
  $y=f(x;\theta)$ using the expression $x=f^{-1}(y;\theta)$ where
  necessary.
\end{prop}

\begin{proof}
  We first check that $P(x;\theta)=f_{x}(x;\theta)\hat{P}(f(x;\theta))$ is
  normalised. To that end, let $y=f(x;\theta)$
\[ \int_{X} P \bd x=\int_{X}
f_{x}\hat{P}\bd x=\int_{X}f_{x}\hat{P}\frac{\bd y}{f_{x}}=1.
\]
 Then we compute the logarithmic derivatives necessary for the Fisher
 information metric
 \begin{equation*}
 \partial_{a}\ln
 P=\frac1{f_{x}\hat{P}(f)}\left(\frac{\bd \hat{P}(f)}{\bd 
     f}(f_{a}f_{x})+\hat{P}(f)(f_{ax})\right).   
 \end{equation*}
 We proceed with the computation by making the change of variables
 $y=f(x;\theta)$
  \begin{align*}
    g_{ab}&=\int_{X}\frac1{(f_{x})^{2}\hat{P}(f)}\left(\frac{\bd \hat{P}(y)}{\bd 
     y}(f_{a}f_{x})+\hat{P}(y)(f_{ax})\right)\left(\frac{\bd \hat{P}(y)}{\bd 
            y}(f_{b}f_{x})+\hat{P}(y)(f_{bx})\right)\bd  y\\
    &=\int_{X}f_{a}f_{b}\frac{\bd  P(y)}{\bd  y}\frac{\bd  \ln P(y)}{\bd 
      y} +
      \frac{f_{ax}f_{bx}}{(f_{x})^{2}}P(y)+\left(\frac{f_{a}f_{bx}+f_{b}f_{ax}}{f_{x}}\right)
      \frac{\bd  P(y)}{\bd  y}\bd  y.
  \end{align*}
Finally, we recognise that $\frac{\partial}{\partial
  x}=f_{x}\frac{\partial}{\partial{y}}$ and that
$f_{a}f_{bx}+f_{b}f_{ax}=\frac{\partial(f_{a}f_{b})}{\partial x}$, and
collect terms to arrive at the result.
\end{proof}

Of course, examining symmetry at such an abstract level cannot be
expected to yield concrete answers immediately and so that the
statement of \cref{prop:symm} is opaque and not obviously useful is
not surprising. Indeed, in what follows we make various simplifying
assumptions about the functional form of the symmetry function $f$ to
arrive at generalisations of familiar results.

We begin by noticing that there is a term in \autoref{eq:symm} which
is proportional to $f_{a}f_{b}$. If it could be arranged that
$f_{a}f_{b}$ be independent of $y$, then we could simply extract a
term proportional to $f_{a}f_{b}$ from the result -- a term whose
importance we already know. Moreover, if we could ensure that the
other terms vanish, we would have $g_{ab}\propto f_{a}f_{b}$ and
achieve our general result once more.

To that end, we choose to require that $f_{x}$ be constant and
$f_{ax}=0$. Although this is likely not the \textit{only} way to
achieve our desired effect, it will certainly suffice. In this case,
we see immediately that $f(x;\theta)=cx+h(\theta)$ is the general
solution -- but this is nothing other than the statement of
translation invariance. Thus, we may achieve the following results by
means of \cref{prop:symm}.

\begin{prop}
  \label{prop:trans}
  Fix a one-dimensional probability density function $\hat{P}$ such
  that the change of variables $y=x-h$ for $h(\theta)$ a
  differentiable function on $M\subset \mathbb{R}^{m}$ leaves the
  spatial domain $X$ unchanged. Let
  $ P(x;\theta)=\hat{P}\left(x-h \right)$ then
  $g_{ab}=(\partial_a h)(\partial_b h)D$ where
  \begin{equation*}
    D=\int_{X}\bd  x \left(\frac{\partial P(x)}{\partial
      x}\right)\left(\frac{\partial \ln P(x)}{\partial x}\right).
  \end{equation*}
\end{prop}

\begin{proof} Apply \cref{prop:symm} to $f(x;\theta)=x-h(\theta)$.
\end{proof}

\begin{corl}
  \label{corl:trans}
Fix one-dimensional probability density functions $P_{i}$ and let
  $h^i(\theta)$ be differentiable on $M\subset \mathbb{R}^{m}$
  and write $y^i=x^{i}-h^{i}$ such that $X_i$ is unchanged under
  this change of variables for all $i$.
  $P(x;\theta)=\bigodot
  P_i\left(x^{i}-h^{i}\right)^{\odot e_i}$
  gives $g_{ab}(P)=(\partial_a h^i)(\partial_b h^j)D_{ij}$
  where
 \[ D_{ij}=\left\{
   \begin{aligned}
     &e_i\int_{X_i}\bd  x^i \left(\frac{\partial P_i}{\partial
      x^i}\right)\left(\frac{\partial \ln P_i}{\partial x^i}\right),&
  i=j\\
  &0,&i\ne j
   \end{aligned}\right. \]
\end{corl}

\begin{proof}
  Combine \cref{prop:trans} and \cref{thrm:sdp}.
\end{proof}

\begin{rema}
  When $P_{i}$ are all Gaussian, $D_{ij}=\delta_{ij}$ and so the
  result of \autoref{eq:cm} follows as a special case.
\end{rema}

To demonstrate how one might achieve the encoding of an arbitrary
Riemannian metric tensor into a spatially disjoint product of
one-dimensional families of probability density functions, consider
the following example.

\newpage

\begin{exmp}
  Suppose we desire a hyperbolic metric tensor $g$ whose associated
  line element is given by
  $\frac1\beta(\bd \alpha^{2}+\bd \beta^{2})$, on the open subset
  $M=\{(\alpha,\beta)\in \mathbb{R}^{2}\,|\,\beta>1\}\subset
  \mathbb{H}^{2} $.
  With some work, it can be shown that an isometric embedding of $M$
  into $\mathbb{R}^{3}$ can be achieved through the function
  \[ h = \left(\frac{\cos \alpha}\beta,\frac{\sin \alpha}\beta,
    \ln\left(\beta+\sqrt{\beta^{2}-1}\right)-\frac{\sqrt{\beta^{2}-1}}\beta\right). 
  \]
  That is, $g=h^{*}\delta$. Moreover, it is evident that $h$ is at
  least $C^{1}$ so we may apply our construction to it and write, for
  example,
\[ P=P_{1}\left(x-h^{1}\right)\odot P_{2}\left(y-h^{2}\right)\odot
P_{3}\left(z-h^{3}\right), \]
  for any one-dimensional probability density functions $P_{i}$ which
  satisfy translation invariance as outlined in \cref{prop:trans}. By
  \cref{corl:trans} we then know that $g(P)=h^{*}D$ and so the
  result follows in the case that $D=\delta$.

  In particular then, we may choose to let $X_{i}=\mathbb{R}$ for
  $i\in\{1,2,3\}$ and put
  \[ \hat P_{1}(x)=\frac1{\sqrt{2\pi}}e^{-\frac12 x^{2}}, \quad
  \hat P_{2}(x)=\frac1\pi\sech x,\quad\hat P_{3}(x)=\frac1{\pi\left(1+x^{2}\right)},\]
  for which $D_{1}=1$ and $D_{2}=D_{3}=\frac12$. Thus, taking the
  values of $D_{i}$ into account, we may write
$P(x,y,z;\alpha,\beta)=\hat P_{1}\left(x-h^{1}\right)\odot
\hat P_{2}\left(y-\sqrt2h^{2}\right)\odot\hat P_{3}\left(z-\sqrt2h^{3}\right)$
to recover \[
P(x,y,z;\alpha,\beta)=
  \frac{ \left(\sqrt{2\pi^{5}}\right)^{-1}
         \sech\left(x-\frac{\sqrt 2\sin \alpha}\beta\right)
         e^{-\frac12\left(y-\frac{\cos \alpha}\beta\right)^{2}} } {1+\left[z+\frac{\sqrt{2\beta^{2}-2}}\beta-\sqrt2\ln\left(\beta+\sqrt{\beta^{2}-1}\right)\right]^{2}},
\]
defined on $\mathbb{R}^{3}\times M$, and for which we know, due to
\cref{corl:trans}, the metric tensor is $g=\beta^{-2}\delta$.
  It may also be verified directly that, given,
  \[
  P_{1}(x;\alpha,\beta) = \frac1{\sqrt{2\pi}}e^{-\frac12\left(x-\frac{\cos
        \alpha}\beta\right)^{2}},\quad
  P_{2}(x;\alpha,\beta) = \frac1\pi \sech\left(x-\frac{\sqrt 2\sin
      \alpha}\beta\right),
  \]\[
 P_{3}(x;\alpha,\beta) = \frac{\pi^{-1}} { 1 +
   \left[x+\frac{\sqrt{2\beta^{2}-2}}\beta -
     \sqrt2\ln\left(\beta+\sqrt{\beta^{2}-1}\right)\right]^{2}}\,\, \textrm{, we have}
  \] 
  \[ g(P_{1})= \frac1{\beta^{4}}
              \begin{bmatrix}
                \beta^{2}\sin^{2}\alpha&\beta\sin\alpha\cos\alpha\\
                \beta\sin\alpha\cos\alpha&\cos^{2}\alpha
              \end{bmatrix},\quad
    g(P_{2})= \frac1{\beta^{4}}
              \begin{bmatrix}
                \beta^{2}\cos^{2}\alpha&-\beta\sin\alpha\cos\alpha\\
                -\beta\sin\alpha\cos\alpha&\sin^{2}\alpha
              \end{bmatrix},
              \] \[
    g(P_{3})=\frac1{\beta^{4}}
              \begin{bmatrix}
                0 & 0\\
                0 & \beta^{2}-1
              \end{bmatrix},
              \] whose sum is as desired -- that is, $g\left(\bigodot
    P_{i}\right)=\sum g(P_{i})$ as \cref{thrm:sdp} assured us. Thus,
  we have managed to encode a desired metric tensor as the Fisher
  information metric of a spatially disjoint product of three,
  one-dimensional families of probability density functions. 
\end{exmp}

We can explore another possible simplifying form of transformation
$f$. Consider that were $f(x;\theta)\propto x$, then every term in
\autoref{eq:symm} would contribute a factor proportional to
$f_{a}$. Again, this is a desirable result and so we explore the
symmetry of scale invariance.

\begin{prop}
  \label{prop:scale}
  Fix a one-dimensional probability density function $\hat{P}$ such
  that the change of variables $y=xe^{h}$ for $h(\theta)$ a
  differentiable function on $M\subset \mathbb{R}^{m}$ leaves the
  spatial domain $X$ unchanged. Let
  $ P(x;\theta)=e^{h}\hat{P}\left(xe^{h} \right)$ then
  $g_{ab}=(\partial_a h)(\partial_b h)E$ where
  \begin{equation*}
    E=\int\limits_XP(x)\left(1+x\frac{\partial\ln P(x)}{\partial
      x}\right)^2\bd  x.
  \end{equation*}
\end{prop}

\begin{proof}
    We set $f(x;\theta)=e^{h(\theta)}x$ and compute the required
  derivatives for \cref{prop:symm} as follows
\[
    f_{a}=\partial_{a}h x e^{h},\quad f_{x} = e^{h},\quad
    f_{ax}=\partial_{a}h e^{h},\quad  \frac{\partial(f_{a}f_{b})}{\partial
                                  y}=2(\partial_{a}h)(\partial_{b}h)y.
 \]
  The result follows straightforwardly.
\end{proof}

\begin{corl}
  \label{corl:scale}
Fix one-dimensional probability density functions $P_{i}$ and let
  $h^i(\theta)$ be differentiable on $M\subset \mathbb{R}^{m}$
  and write $y^i=x^{i}e^{h^{i}}$ such that $X_i$ is unchanged under
  this change of variables for all $i$.
  $P(x;\theta)=\bigodot
  e^{h^{i}}P_i\left(x^{i}e^{h^{i}}\right)^{\odot e_i}$
  gives $g_{ab}(P)=(\partial_a h^i)(\partial_b h^j)E_{ij}$
  where
 \[ E_{ij}=\left\{
   \begin{aligned}
     &e_i\int_{X_i}\bd x^i P^i \left(1+x\frac{\partial\ln P_i}{\partial
      x^i}\right)^2,&
  i=j\\
  &0,&i\ne j
   \end{aligned}\right. \]
\end{corl}

\begin{proof}
  Combine \cref{prop:scale} and \cref{thrm:sdp}.
\end{proof}

\begin{corl}
  Every Riemannian metric tensor may be reached as the result of the
  Fisher information metric acting upon a spatially disjoint product
  of families of one-dimensional probability density functions.
\end{corl}

\begin{proof}
  Apply either \cref{corl:scale} or \cref{corl:trans} to the desired
  $C^{1}$ pullback $h$, which exists due to the isometric embedding of
  the desired manifold in $\mathbb{E}^{n}$ via the Nash Embedding
  theorem.
\end{proof}

It can now be seen that relatively simple computations give rise to
highly useful results by way of \cref{thrm:sdp}. Indeed, to extend
this work one need only find other families of probability density
functions whose Fisher information metric can be made to be
proportional to $(\partial_{a}h) (\partial_{b}h)$ in order to combine
them in the requisite multiplicity to allow $h$ to be the pullback for
a desired Riemannian metric tensor. That we made explicit use of
spatial domain symmetries using \cref{prop:symm} should be seen as
merely a convenient and intuitive way of making use of \cref{thrm:sdp}
to construct desirable results.

\section{Discussion}
That we can associate a Reimannian information manifold with a well-defined 
metric to a given family of probability distribution functions is a remarkable
thing. Indeed, the power of this statement immediately begs the question of how much statistical, or information theoretic properties can be captured in the language of differential geometry. It is clear that the Fisher metric captures only a small amount of information about the family of PDFs, however the metric is but one differential geometric structure, and one could imagine that more information may be translated into the language of form fields of different order.

What we have shown here is in line with the string theory ideas of  holographic duality, which indicate that any scale-free gauge theory should give rise to a hyperbolic geometry. Different scale-free gauge theories should however give rise to different field contents, above and beyond the metric, depending on the operators which can be formed in the gauge theory. As discussed in the introduction, information geometry has already been used to go from: $\text{gauge theory}\rightarrow\text{PDF}\rightarrow\text{metric}$. Thus it would be interesting, both from the information theoretic point of view, as well as from the holographic point of view to see what more differential structure can be encoded in such mappings.

This article is our attempt to formulate a crisp statement about the uniqueness of the association of a metric to a probability distribution. We saw how the Fisher information metric took a spatially disjoint product of probability distributions to a sum of the individual metric tensors. We leveraged this result to entirely reverse the computation, in generality. 
In fact, we found that it is possible to explicitly construct any Riemannian metric 
via the spatially disjoint product of one-dimensional probability density functions exhibiting a select spatial domain symmetry. This symmety in fact features in a crucial way in 
our construction to inject dependence upon the
components of the pullback used to isometrically embed the desired
metric in $\mathbb{E}^{n}$. Moreover, up to the spatial domain
symmetries mentioned and some mild conditions on the continuity of the
probability density functions, we have shown that such a construction may be
given in terms of arbitrary probability density functions. 

While our results appear to be quite negative in terms of the amount of information encoded in the Fisher metric from a PDF, we propose to interpret it as a signal that, in order to fully capture a duality that seems to point to a one-to-one map between string theory on $AdS_{5}\times S^{5}$ and maximally supersymmetric Yang-Mills theory on the $AdS$ boundary, a deeper understanding of information geometry is required. We leave this for future work.
 
\section{Acknowledgements}
JS and TC are grateful for the URC National Research Foundation (NRF)
of South Africa under grant number 87667. JM acknowledges support from
the NRF Competitive Support for Rated Researcher program under grant
CPRR 90519.

\end{document}